%% file: main.tex
\documentclass[aps,prl,twocolumn,superscriptaddress,longbibliography,nofootinbib]{revtex4-2}
\usepackage[T1]{fontenc}
\usepackage[utf8]{inputenc}
\usepackage{color}
\usepackage{xcolor}
\usepackage{graphicx}
\usepackage{babel}
\usepackage{amsmath}
\usepackage{amsthm}
\usepackage{amsfonts}
\usepackage{mathrsfs}
\usepackage{dsfont}
\usepackage{bm}
\usepackage{stmaryrd}
\usepackage{physics}
\usepackage{quantikz}
\usepackage[unicode=true,pdfusetitle,
bookmarks=true,bookmarksnumbered=false,bookmarksopen=false, breaklinks=true,pdfborder={0 0 0},pdfborderstyle={},backref=false,colorlinks=true]{hyperref}
\usepackage{placeins}

\usepackage{lipsum}

\newtheorem{theorem}{Theorem}

\newtheorem{lemma}{Lemma}

\begin{document}
\title{Probing quantum advantage for solving the Fermi-Hubbard model with entropy benchmarking}

\author{Pauline Besserve}
\affiliation{School of Informatics, QSL, University of Edinburgh, 10 Crichton Street, Edinburgh EH8 9AB, United Kingdom}
\author{Ra\'ul Garc\'ia-Patr\'on}
\affiliation{School of Informatics, QSL, University of Edinburgh, 10 Crichton Street, Edinburgh EH8 9AB, United Kingdom}
\affiliation{Phasecraft Ltd., London, United Kingdom}

\begin{abstract}
We developed a practical quantum advantage benchmarking framework that connects the accumulation of entropy in a quantum processing unit and the degradation of the solution to a target optimization problem. The benchmark is based on approximating from below the Gibbs states boundary in the energy-entropy space for the application of interest. We believe the proposed benchmarking technique creates a powerful bridge between hardware benchmarking and application benchmarking, while remaining hardware-agnostic. It can be extended to fault-tolerant scenarios and relies on computationally tractable numerics. We demonstrate its applicability on the problem of finding the ground state of the two-dimensional Fermi-Hubbard.
\end{abstract}
\maketitle

The last decade has witnessed an impressive improvement in the performance and size of quantum processing units on a wide variety of hardware platforms \cite{wintersperger_neutral_2023, acharya_quantum_2025_compressed, kim_evidence_2023-2, xu_constant-overhead_2024, aghaee_rad_scaling_2025_corrected}. This has motivated a surge of activity on algorithms with applications in key topics such as quantum chemistry \cite{berry_rapid_2024}, many-body physics \cite{stroeks_solving_2024, yu_sample-based_2025}, machine learning \cite{peral-garcia_systematic_2024_corrected}, and optimization \cite{jordan_optimization_2025}. Experimental implementations of non trivial random quantum circuits with up to hundreds of qubits \cite{liu_certified_2025} or optical modes \cite{zhong_quantum_2020, madsen_quantum_2022} have claimed to reach quantum advantage, in spite of errors and imperfections in the hardware. 
If it is common consensus that a demonstration of quantum advantage on a problem of scientific, industrial or commercial interest has not happened yet, the most recent hardware and quantum error-correction improvements \cite{bravyi_high-threshold_2024, paetznick_demonstration_2024, acharya_quantum_2025_compressed} have significantly shifted the question  from "Would it happen at all?" to "When will it happen?". 

Against this background, it is important to find ways of assessing overall progress toward quantum advantage, while identifying the most viable pathways.
This has motivated efforts to benchmark quantum applications \cite{proctor_benchmarking_2025} and to design
performance scores for families of problems with potential for quantum advantage \cite{martiel_benchmarking_2021, wu_variational_2024}.
Ideally, one would like to develop a framework with the following features: 1) addresses a large family of relevant problems containing expected candidates for practical quantum advantage; 2) is hardware agnostic; 3) remains applicable as hardware evolves from near-term to fault-tolerant architectures; 4) whose classical post-processing is computationally tractable.

Building on previous work \cite{stilck_franca_limitations_2021}, we develop a formalism exhibiting all those features that
targets a wide family of relevant problems. These range from the search of the ground-state of many-body Hamiltonians to the solution of a classical optimization problem, i.e., any problem that can be recast as the minimization of a cost function defined as the expectation value $E(\rho)={\rm Tr}(\rho H)$ of a given Hamiltonian $H$. The formalism connects the accumulation of entropy on the working Quantum Processing Unit (QPU) to the degradation of the quality of the solution. In a nutshell, it exploits the well-known (quantum) statistical mechanics fact that for a given energy (expected cost function
for $\rho$) the Gibbs state maximizes the von-Neumann entropy \cite{jaynes_information_1957, jaynes_information_1957-1}. 
As sketched in Figure \ref{fig:bm_method}, Gibbs states separate the energy and entropy parameter space $(E,S)$ into physically achievable and unachievable domains (light red-shaded area). 
Therefore, an energy $E_{\mathrm{class}}$ achieved by a classical solver can be translated into a threshold of entropy $S_{\mathrm{th}}$ beyond which quantum advantage is no longer achievable, as any potential QPU that operates beyond it can only provide a worse solution. 

\begin{figure}[h!]
    \centering
    \scalebox{0.55}{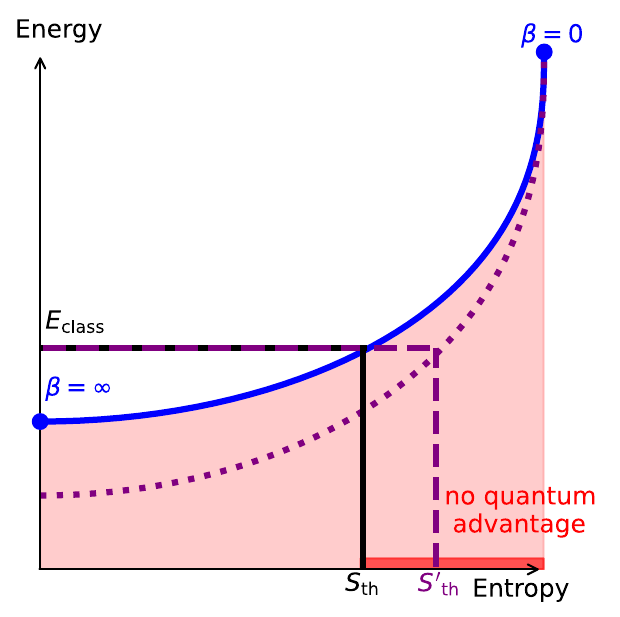}
    \caption{ \textbf{Gibbs states entropy benchmarking framework.} Gibbs states (blue line) separate the energy-entropy parameter space $(E,S)$ into a region of physically achievable parameters from unachievable ones (light red-shaded area). A classical solver providing a candidate solution $E_{\mathrm{class}}$ leads to an entropy threshold $S_{\mathrm{th}}$.
    Any quantum algorithm running on a realistic QPU surpassing $S_{\mathrm{th}}$ (bright red area) cannot provide quantum advantage. A lower-bound to the Gibbs state boundary (purple) can only yield a more conservative benchmark ($S'_{\mathrm{th}} \geq S_{\mathrm{th}}$).}
    \label{fig:bm_method}
\end{figure}

It is known that computing the Gibbs state boundary can only be at least as hard as minimizing ${\rm Tr}(\rho H)$, as the zero-temperature solution includes the solution to the minimization problem. 
To circumvent this critical bottleneck, we developed a relaxation of the original Gibbs boundary by
finding lower-bounds that are efficiently computable. An entropy threshold derived from a lower-bound to the Gibbs state boundary ($S'_{\mathrm{th}},$ in Figure \ref{fig:bm_method}) will be higher than the actual threshold resulting from the Gibbs boundary itself. Therefore, any subsequent prediction would be conservative (generous to the QPU performance) but correct. When possible, we will aim to derive bounds in terms of intensive quantities, like energy and entropy densities,  opening the door to energy-entropy lower-bounds that are size-independent and need to be computed only once for a large family of problem instances.

We believe the proposed benchmarking technique creates a powerful bridge between hardware benchmarking and application benchmarking.  Three steps are to be followed to assess the capability of a variety of existing hardware platforms to solve a given optimization problem. Firstly, using our framework one can obtain an efficiently computable energy vs entropy lower-bound for the problem of interest. Secondly, 
exploiting hybrid numerical and hardware implementation techniques \cite{demarty_entropy_2024}, one builds heuristic models of how entropy accumulates in quantum circuits for the platforms of interest. Equivalently, one can use models provided by a trusted third party or vendor. Ideally, those models would be characterized by a restricted set of relevant parameters, like number of qubits, number of gates, gate calibration data, and, if necessary, key information about the architecture topology. Finally, combining the knowledge of the best state-of-the art classical solver and the energy-entropy lower-bound obtained before, one can estimate an entropy density threshold beyond which quantum advantage is guaranteed to be lost for the problem of interest. This information combined with the model of entropy accumulation for a given QPU, can be used to provide a bound on the circuit volume accessible to solve the problem of interest before quantum advantage is guaranteed to be lost. 

A big advantage of this approach is that it creates an efficient "separation of labor" between the QPU benchmarking (building an entropy accumulation model for a QPU \cite{demarty_entropy_2024}) and application benchmarking (building the energy-entropy boundary for a given problem). The entropy accumulation may potentially depend on the circuit family topology or the ansatz family in VQE,  but is independent of the problem, making significant savings on QPU use. Similarly, the application energy-entropy analysis and choice of classical solver are totally independent of the choice of hardware. Moreover, the former is computationally cheap and does not require quantum computation expertise to be carried out by a potential end-user. An additional advantage is the capability of the framework to make predictions for architectures before fabrication, if we have a trusted model of its entropy accumulation at our disposal. 

In what follows, after formally introducing our framework, we illustrate this general benchmarking technique on the problem of finding the ground state of the two-dimensional Fermi-Hubbard \cite{qin_hubbard_2022}, a problem that is considered to be a good candidate for quantum advantage \cite{fauseweh_quantum_2024}. 
The methodology is easily transferable to other optimization problems. 

\paragraph{\textbf{Benchmarking framework---}}
Let $H$ be some Hamiltonian operator acting on a Hilbert space $\mathscr{H}$ of dimension $2^N$, and $\mathcal{D}$ denote the set of density matrices. We are interested in the following minimization problem:
\begin{equation}
    E^* = \min_{\rho \in \mathcal{D}} E(\rho) = \min_{\rho \in \mathcal{D}} \mathrm{Tr}(\rho H) 
\end{equation}
where the task of preparing $\rho$, measuring the value of the cost function $ \mathrm{Tr}(\rho H)$, but also potentially updating the state, is delegated to a quantum computer. This family covers many implementations of hybrid algorithms designed for the NISQ era, such as the Variational Quantum Eigensolver (VQE \cite{tilly_variational_2022}) and the Quantum Approximate Optimization Algorithm (QAOA \cite{farhi_quantum_2014}). 

The benchmark relies on the von-Neumann entropy content of the state prepared by the quantum computer, which reads $S(\rho)=-\mathrm{Tr}(\rho \ln\rho)$, where we use the notation $\ln$ for the logarithm in base 2.
The Gibbs state relative to $H$ at inverse temperature $\beta$ is defined as
\begin{equation}
    \sigma_{\beta} = \frac
    {e^{-\beta H}}{\mathrm{Tr}\left( e^{-\beta H} \right)} = \frac
    {e^{-\beta H}}{Z_{\beta}}
\end{equation}
where $Z_{\beta}$ is called the partition function. 
It is a well-know fact in (quantum) statistical mechanics that within all states with the same energy, i.e., for all $\sigma$ satisfying $\mathrm{Tr}(\rho H) = E = \mathrm{Tr}(\sigma_{\beta_E} H)$, the Gibbs state maximizes the entropy: $S(\sigma_{\beta_E})\geq S(\rho)$. Equivalently, its dual form states that for all $\rho$ satisfying $S(\rho) = S = S(\sigma_{\beta_S})$, the Gibbs state minimizes the energy $ \mathrm{Tr}(\rho H) \geq \mathrm{Tr}(\sigma_{\beta_S} H)$.  This directly implies that Gibbs states are the boundary of all physically reachable points in the parameter space $(E,S)$, as sketched in Figure \ref{fig:bm_method}. For completeness, we give a detailed proof using information theory tools in Appendix \ref{app:gibbs_boundary}. 

In Appendix \ref{app:bounds_equiv}, we prove that the seemingly alternative energy-entropy boundary of the parameter space $(E,S)$ derived in equation (12) of Reference \cite{stilck_franca_limitations_2021} is nothing else than a reformulation of the Gibbs state boundary presented here. In a nutshell,  equation (12) 
in \cite{stilck_franca_limitations_2021} can be recast as the union of all tangents to the family of Gibbs states. The advantage of our new perspective is to open the way to novel and potentially more efficient relaxations of the Gibbs state boundary. 

The framework presented here  can be trivially extended to algorithms implemented within quantum error-correction, where the role of the entropy of all physical qubits is assumed by the entropy of the logical qubits. 

\paragraph{\textbf{Lower-bounds methodology---}}

Determining the lieu of Gibbs states points $(e(\sigma_{\beta}), s(\sigma_{\beta}))$ is at least as hard as determining the ground state energy, since $\sigma_{\infty}$ is the ground state. 
Our formalism opens the door to deriving tractable lower bounds to the Gibbs state boundary that perform tighter than a simple extrapolation of the high-temperature regime as used in \cite{stilck_franca_limitations_2021}. A simple way to get a lower-bound to a Gibbs state energy is to decompose the Hamiltonian $H$ into some sum of different Hamiltonians $H_j$ for which the Gibbs state problem is individually tractable. As proven in Appendix \ref{app:lemma_partition}, then, for any quantum state $\rho$ such that $S(\rho)=S$, we have 
\begin{equation}
    \mathrm{Tr}(\rho H) \geq \sum_j\mathrm{Tr}(\sigma_{\beta_j}^{(j)} H_j)
    \label{eq:lower_bounding_recipe}
\end{equation} 
where each Gibbs state $\sigma_{\beta_j}^{(j)}$ is paired to a Hamiltonian $H_j$ while satisfying the constraint $S(\sigma_{\beta_j}^{(j)}) = S$. 
We note that this technique generalizes that employed in Reference \cite{valenti_lower_1991} which was only concerned with ground state energies.
In what follows, we illustrate this idea on the case of the two-dimensional Fermi-Hubbard model (abbreviated FHM in the remainder of the text) \cite{hubbard_electron_1997}.

\paragraph{\textbf{Application to the Fermi-Hubbard model---}}
\label{sec:app_FH}
Finding the ground state of the Fermi-Hubbard Model (FHM) is a problem of scientific and practical interest for which we have strong evidence of being intractable on classical computers and
is believed to be a good candidate to provide quantum advantage. This ubiquitous model allows to draw very rich phase diagrams for strongly correlated systems, guiding material design for a wide variety of applications, ranging from the manufacturing of smart materials to the highly desirable manifestation of room-temperature superconductivity \cite{adler_correlated_2018, arovas_hubbard_2022}. 
The FHM is defined  on a graph $G = (V, E)$ as 
\begin{equation}
    H = -t\sum_{
    \substack{(i, j) \in E  \\ \sigma=\uparrow, \downarrow} 
    }(c^{\dagger}_{i\sigma} c_{j \sigma} + \mathrm{h.c.}) - \mu \sum_{
    \substack{i \in V, \\
    \sigma=\uparrow, \downarrow}
    }n_{i\sigma} + U\sum_{i \in V}n_{i\uparrow}n_{i \downarrow}.
\end{equation}
We consider a square lattice comprising $L$ sites in each direction. Such a system hosts a number $N=2L^2$ of fermionic modes. We make the choice of periodic boundary conditions (abbreviated PBC in what follows) in order to ensure translational invariance, which is relevant to the study of crystal lattices in the thermodynamic limit. 
The FHM has trivial solutions for either $U/t=0$ (non-interacting limit) or $U/t = \infty$ (atomic limit). 

As illustrated for $L=4$ on Figure \ref{fig:tiling_decomp}, in what follows, we explore three partitioning approaches to obtain Gibbs state boundary lower-bounds which can be divided into two conceptual families. 
\begin{figure}[h]
    \centering
    \includegraphics[width=0.85\linewidth]{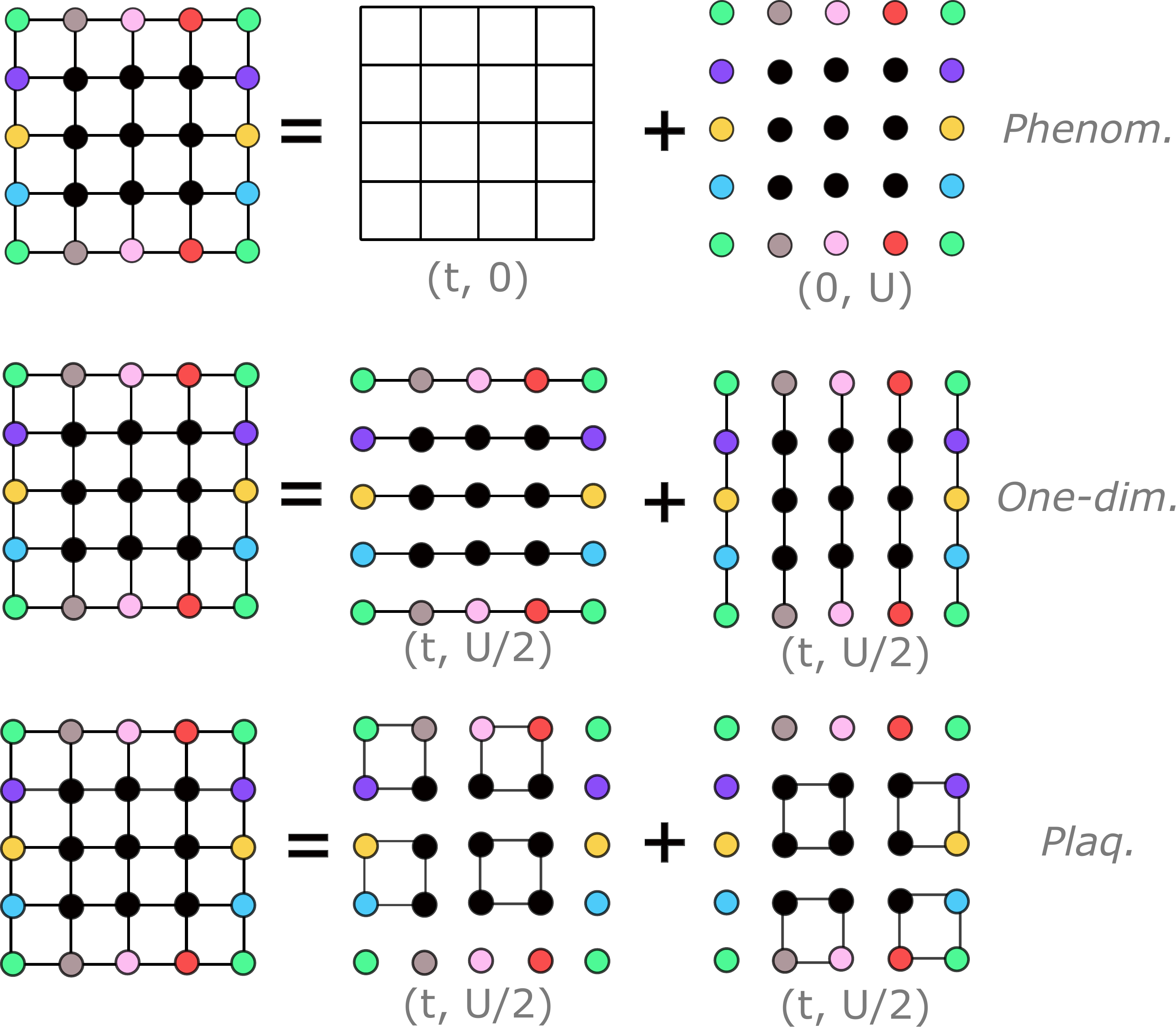}
    \caption{\textbf{Partitioning techniques.} The three different decompositions of the periodic 2D FHM used in this work to compute lower bounds to the Gibbs state energies. Example of a size $L=4$ model. PBC are materialized with colored dots representing the same site. Phenomenological partitioning (\textit{Phenom.}) between the kinetic and the atomic terms, and geometric partitionings into respectively 1D Fermi-Hubbard systems (\textit{One-dim.}) and Fermi-Hubbard plaquettes (\textit{Plaq.}, valid for even values of $L\geq 4$).}
    \label{fig:tiling_decomp}
\end{figure}

\paragraph{Phenomenological partitioning (\textit{Phenom.})}
The first partitioning addresses separately the kinetic part, namely the tight-binding Hamiltonian $H_{\mathrm{TB}} = -t\sum_{\langle i, j \rangle \in E, \sigma=\uparrow, \downarrow} (c^{\dagger}_{i\sigma} c_{j \sigma} + \mathrm{h.c.})$, and the interaction part $H_{\mathrm{at}} = - \mu \sum_{i \in V} (n_{i\uparrow} + n_{i \downarrow}) + U \sum_{i \in V}n_{i\uparrow}n_{i \downarrow}$. Each of these two Hamiltonians is tractable as they both correspond to non-interacting fermionic modes, respectively in the Fourier basis and in the original basis of localized orbitals, as expanded in Appendix \ref{app:TB_atomic}.  

\paragraph{Geometric partitionings} We then consider partitionings that consist in the design of tilings composed of independent, tractable units which sum to our initial Hamiltonian $H$. This approach is reminiscent of that of Anderson to compute a lower-bound to the ground state energy of lattice Hamiltonians \cite{anderson_limits_1951, eisert_note_2023}   and can be seen as an adaptation of the method to non-zero temperatures. By geometric we mean that each group of terms $H_j$ will effectively act onto a subset of the lattice, however one lattice site can be acted on by several groups.  Our first example (\textit{One-dim.}) consists in separating the Hamiltonian into a horizontal component $H_{\mathrm{h}}$ and a vertical component $H_{\mathrm{v}}$, each corresponding to 1D systems which aren't interacting between each other (see details in Appendix \ref{app:LB2}). 
The second example (\textit{Plaq.}) consists in decomposing a $L \times L$ model with PBC as a sum of two Hamiltonians acting on the sites with disconnected plaquette ($2\times2$) Hamiltonians with open boundary conditions (OBC), provided $L$ is even and greater than 4 . The details are provided in Appendix \ref{app:LB3}.

\begin{figure}
    \centering
    \includegraphics[width=\columnwidth]{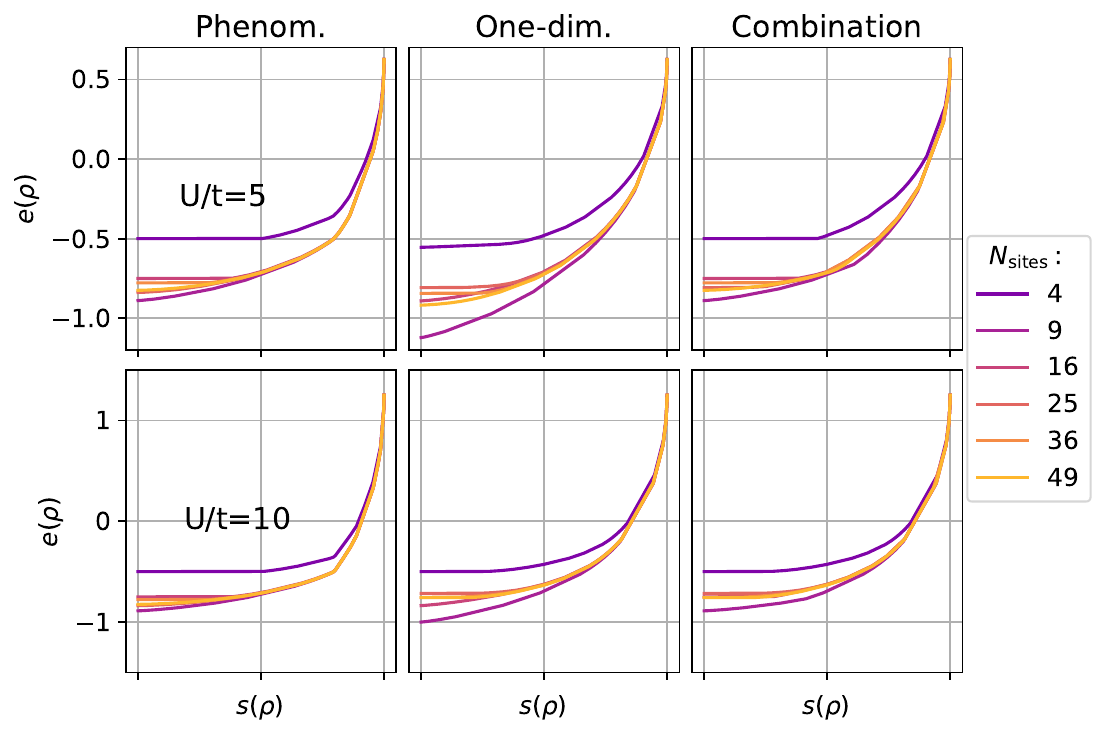}
    \caption{\textbf{Onset of scale-invariance for the Gibbs state energy lower-bounds.} Results are presented in challenging regimes of correlation $U/t=5, 10$ for the \textit{Phenom.} and \textit{One-dim.} lower-bounds, as well as the combination obtained upon taking locally the best lower-bound among the three investigated \textit{Phenom.}, \textit{One-dim.} and \textit{Plaq.}. The latter lower bound only applies to cases of even $L\geq4$ and is not shown here, as scale invariance is inherent in this case.}
    \label{fig:onset_scale_invariance}
\end{figure}

\begin{figure}
    \centering
    \includegraphics[width=\columnwidth]{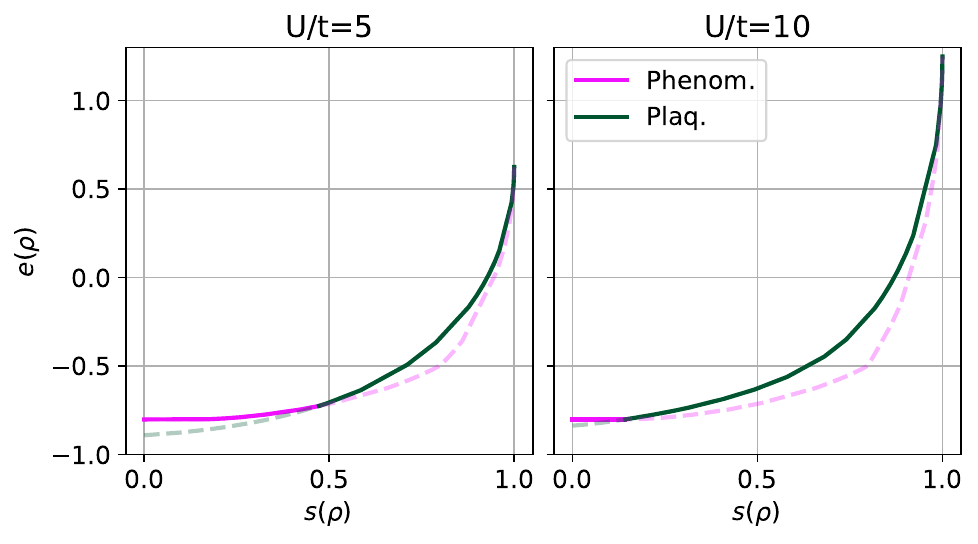}
    \caption{\textbf{Behaviour of the combination of lower-bounds.} Best lower bound to the Gibbs states energy density as a function of the entropy density for $N_{\mathrm{sites}}=144$, for different values of the correlation $U/t$. The dashed lines materialize the continuation of each lower-bound.}
    \label{fig:best_LB}
\end{figure}

\paragraph{\textbf{Energy-Entropy bounds for FHM---}}
Here we present lower-bounding results for $\mu=0$ for different levels of correlation $U/t$. Energy-entropy (densities) lower bounds that are size independent have the advantage of settling the study of a whole problem family with a single evaluation. While scale invariance is guaranteed by design for \textit{Plaq.}, we can numerically observe on Figure \ref{fig:onset_scale_invariance} that it arises as the size increases for \textit{Phenom.} and \textit{One-dim.}.
In the high-entropy part of the diagram, size invariance arises very quickly (at size $N_{\mathrm{sites}}=16$) whereas in the low entropy regime, although it manifestly does not get attained, we observe a concentration of the energy density past this size. We see that we can consider that scale invariance was attained for all three lower bounds for the largest size $N_{\mathrm{sites}}=49$.
An additional improvement is to resort to the union of several lower-bounding techniques, although crude, in order to enhance the tightness of their collective lower-bound ('\textit{Combination}' lower-bound on Figure \ref{fig:onset_scale_invariance}, defined locally as  the largest value among the different lower bounds). Scale invariance is preserved by the combination. 

We then proceed to study sizes beyond $N_{\mathrm{sites}}=49$, up to $N_{\mathrm{sites}}=144$. Since \textit{One-dim} is not currently scalable to large sizes, we
fuse the lower bounds \textit{Phenom.} and \textit{Plaq.} when available. Results for  $N_{\mathrm{sites}}=144$ are presented on Figure \ref{fig:best_LB}.

We see that as $U$ increases, \textit{Plaq.} starts yielding higher lower-bounds on a larger and larger portion of the high-entropy part of the diagram. We do expect \textit{Phenom.} to be a good lower bound in the low-$U$ regime since the Coulomb term can then be neglected, which is confirmed. On the other hand, in the high-$U$ regime it is expected that both the \textit{Phenom.} and \textit{Plaq.} lower-bounds both become tight as the kinetic term becomes negligible. It is thus particularly suitable to combine lower bounds in the intermediate regime.

\paragraph{\textbf{Benchmarking quantum advantage---}}
 Finally, we consider a real-case scenario for our benchmarking technique. We consider the half-filled FHM on a square lattice of size $8 \times 8$ , in the non-trivial correlation regime $U/t = 4$. In a nutshell, this size escapes exact diagonalization  whereas intermediate correlation means both perturbations around the non-interacting case and tensor network methods will prove inefficient. Using Jordan-Wigner encoding, it can be addressed by a 128-qubit platform, a size compatible with NISQ hardware. It has the additional advantage of having been well studied by classical methods—producing a rich body of literature for comparison—while also drawing considerable attention from quantum algorithm research. On the other hand, it is known that Monte-Carlo methods  \cite{white_numerical_1989} such as AFQMC \cite{zhang_quantum_2003} are numerically exact in the half-filled case, as they do not suffer from a sign problem. Therefore, we recognize that finding the ground state at half-filling is not a suitable candidate for demonstrating quantum advantage. Nonetheless, we consider it a relevant benchmarking case, as it is reasonable to expect that a first step toward quantum advantage in optimization (e.g. away from half-filling) will be the reproduction of classically established results. We discuss with further detail this choice in Appendix \ref{app_sub:choice_pb}.

According to the literature \cite{leblanc_solutions_2015, qin_benchmark_2016} there is a good agreement between different methods to compute the energy density in the ground state, leading to $e_{\mathrm{class}}=-1.43$. Using the best tractable lower-bound, which is \textit{Plaq.}, we get graphically an entropy density threshold $s_{\mathrm{th}}=0.69$, as detailed in Appendix \ref{app:benchmarking_details}.

We are now ready to estimate the maximum quantum circuit depth we can afford when preparing a target ground state solution to the FHM before we are guaranteed to have lost any potential quantum advantage. 
Based on accumulated evidence against the use of hardware-efficient ansatze for VQE due to trainability issues \cite{wang_noise-induced_2021, cerezo_cost_2021}, we consider variational ground state preparation with two distinct physically motivated ansatz circuits commonly used to tackle the FHM, namely the Low-Depth Circuit Ansatz (LDCA \cite{dallaire-demers_low-depth_2019}) and the Hamiltonian Variational Ansatz (HVA \cite{wecker_towards_2015}). 
 
In Figure \ref{fig:max_layers}, we provide the results obtained using the heuristic model of entropy accumulation under depolarizing noise devised in Reference \cite{demarty_entropy_2024}, with varying two-qubit depolarizing probability $p_2$.

In spite of our benchmark being very conservative and the high entropy threshold, we obtain a rather negative result for current accessible noise levels. 
For example, the best depolarizing probabilities of around $p_2 = 3.10^{-4}$ for trapped-ions devices \cite{loschnauer_scalable_2024}, merely delineates the start of a regime in which one can at least run one layer of each circuit (see inset of Figure \ref{fig:max_layers}). 

Assessing a potential quantum advantage also requires theoretical insight into the required circuit depth in order to explore a portion of the Hilbert space likely to contain the target ground state, which is a whole challenge in itself. In challenging regimes of intermediate $U/t$ it is nevertheless not expected that a single-layer LDCA or HVA circuit can prepare the ground state of a 2D FHM. Empirical evidence for small instances up to $L=5$ rather suggests a scaling $N_{\mathrm{layers}} \sim L^2$ (the number of sites) for the HVA ansatz \cite{cade_strategies_2019}. This means we would need a two-qubit depolarizing probability $p_2 \sim 2.10^{-5}$, which is closer to what is expected in the fault-tolerance regime than within a near-term experiment.

\begin{figure}
    \centering
    \includegraphics[width=0.8\columnwidth]{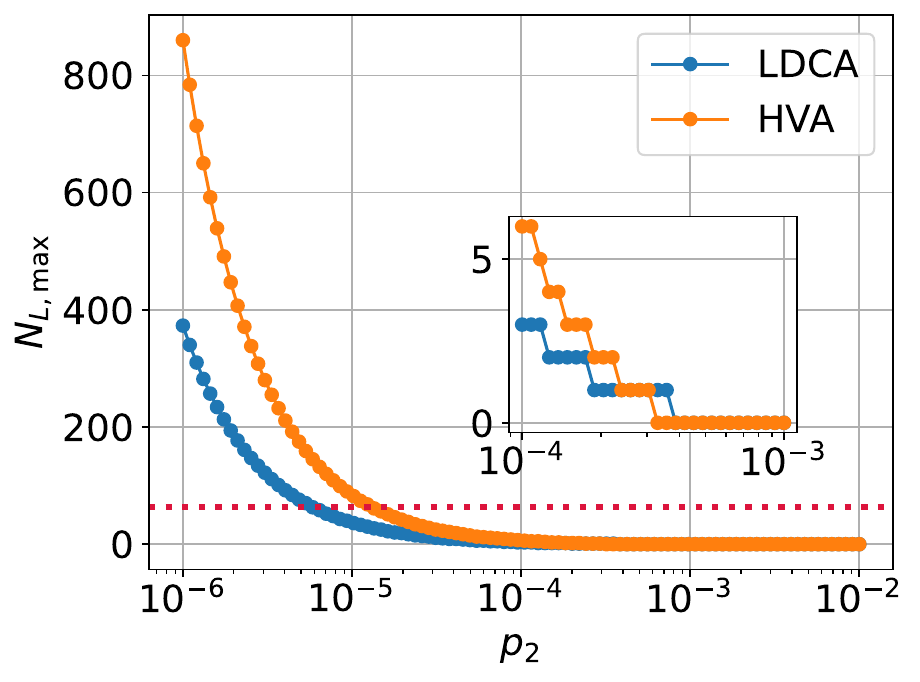}
    \caption{Maximum number of layers allowed for the LDCA and HVA circuit ansatze in order to remain below the entropy density threshold above which classical superiority is certified.  The considered application is the ground state preparation of the 8 $\times$ 8 half-filled Hubbard model in 2D in the challenging intermediate regime of correlation $U/t = 4$. The dotted red line materializes the expected required depth for the HVA circuit to explore a suitable portion of the Hilbert space.}
    \label{fig:max_layers}
\end{figure}

\paragraph{\textbf{Conclusion---}} 
We developed a practical entropy benchmarking framework requiring few classical computational resources to shed light on the suitability of a quantum computing approach to solve a given optimization task. The benchmark is based on locating or lower-bounding the Gibbs states frontier in the energy-entropy space for the application of interest. Our results show that albeit simple, our entropy benchmarking approach yields no-go results for the solving of the two-dimensional Fermi-Hubbard model at large sizes using NISQ hardware.

However, we assumed here that raw, unmitigated results were to provide the target energy. 
It will be the object of future work to incorporate error mitigation \cite{endo_hybrid_2020} into our framework and determine whether it can modify the conclusion, or is limited by a prohibitive shot count \cite{quek_exponentially_2024}. It should nonetheless be stressed that the entropy benchmarking framework intrinsically does not take into account other pitfalls of the VQE approach, such as the trainability of the circuit. 

Our work raises interesting theoretical questions on how to best lower-bound the Gibbs states energy-entropy boundary, as a better bound will give more accurate and also more stringent constraints on the performance of quantum computers. In practice, more refined numerical techniques to compute the Gibbs states boundary can lead to stronger benchmarking results. In this line the \textit{One-dim.} lower-bound could be used for bigger instances than illustrated here by resorting to tensor-network methods, which were shown to be particularly efficient for one-dimensional systems \cite{white_density_1992}. 

Finally, we only considered here the direct solving of the Fermi-Hubbard model. Embedding methods such as DMFT \cite{georges_dynamical_1996} consider a simpler, proxy model to Fermi-Hubbard. Despite the limitations we presented for the direct approach to the FHM, hybrid quantum-classical implementations of embedding schemes could still provide quantum advantage in the late-NISQ era \cite{chen_quantum-classical_2025, bertrand_turning_2025}. 
An adaptation of our benchmarking framework to those algorithms and other hybrid quantum-classical approaches is also an interesting question for further investigation.

\paragraph{\textbf{Acknowledgments---}}
The authors were supported by the EPSRC-funded project Benchmarking Quantum Advantage.

\paragraph{\textbf{Data availability---}}The code and data are available at \url{http://github.com/pbesserve/BQA_Hubbard}.

\FloatBarrier

\input{BQA_Hubbard.bbl}
\appendix 

\section{Gibbs state boundary}
\label{app:gibbs_boundary}

\begin{theorem}
    Consider the energy-entropy space $(U,S)$ (as depicted on Figure \ref{fig:bm_method} of the main text), where $U(\rho)=\mathrm{Tr}[\rho H]$ and the von Neumann entropy reads $S(\rho)=-\mathrm{Tr}[\rho\log\rho]$. Among all physically reachable points in the space $(U,S)$, the family of Gibbs state $\sigma_\beta=e^{-\beta H}/Z_\beta$, where $Z_\beta=\mathrm{Tr}[e^{-\beta H}]$, with free parameter the inverse temperature $\beta$ is the lower-bound. Its primal formulation states that among all states $\rho$ with same expectation energy $\mathrm{Tr}[\rho H]=E$, including the Gibbs state $\sigma_{\beta_E}=e^{-\beta_E H}/Z_{\beta_{E}}$, the entropy is maximized by the Gibbs states, i.e., $S(\sigma_\beta)\geq S(\rho)$. The dual formulation states that among all 
    all states $\rho$ with same entropy $S(\rho)=S$ including the Gibbs state 
    $\sigma_{\beta_S}=e^{-\beta_S H}/Z_{\beta_{S}}$, the energy is minimized by the Gibbs state,
    i.e., $\mathrm{Tr}[\rho H]\geq \mathrm{Tr}[\sigma_{\beta_S} H]$.
\end{theorem}

\begin{proof}
    \textbf{Primal}. Using the relative entropy $D(\rho||\tau)=\mathrm{Tr}[\rho(\log\rho-\log\tau)]$ we can write 
    \begin{align}
        D(\rho||\sigma_{\beta_E}) &= -\mathrm{Tr}[\rho\log\sigma_{\beta_E}]-S(\rho) \\
        &= -\mathrm{Tr}[\rho(-\beta_E H-\log Z_{\beta_E})]-S(\rho) \\
        &= \beta_E\mathrm{Tr}[\rho H]+\log Z_{\beta_E}-S(\rho)\geq 0,
    \end{align}
    where we used the definition of relative entropy, Gibbs state and the positivity of the relative entropy. It is well-know that $ D(\rho||\rho)=0$ for any $\rho$, therefore $D(\sigma_\beta||\sigma_{\beta_E})=0=\beta\mathrm{Tr}[\rho H]+\log Z_{\beta_E}-S(\sigma_{\beta_E})$. 
    Subtracting both equation we obtain 
    \begin{align}
      D(\rho||\sigma_{\beta_E})-D(\sigma_\beta||\sigma_\beta)=\beta_E\left(\text{Tr}[\rho H]-\text{Tr}[\sigma_{\beta_E} H]\right)  \nonumber\\ 
       +S(\sigma_\beta)-S(\rho)\geq 0,
    \end{align}
    which using the constraint $\text{Tr}[\rho H]=\text{Tr}[\sigma_\beta H]$, leads to 
    $S(\sigma_{\beta_E})-S(\rho)\geq 0$. \newline
    
    \textbf{Dual} Similarly as above replacing $\sigma_{\beta_E}$ by $\sigma_{\beta_S}$ one can obtain
\begin{align}
        D(\rho||\sigma_{\beta_S}) &= \beta_S\mathrm{Tr}[\rho H]+\log Z_{\beta_E}-S(\rho)\geq 0,\\
        D(\sigma_{\beta_S}||\sigma_{\beta_S}) &= \beta_S\mathrm{Tr}[\sigma_{\beta_S} H]+\log Z_{\beta_S}-S(\sigma_{\beta_S})=0.
    \end{align}
    Because by definition we have $S(\sigma_{\beta_S})=S(\rho)$ the difference of both equation leads to $\mathrm{Tr}[\rho H]\geq \mathrm{Tr}[\sigma_{\beta_S} H]$.
\end{proof}

\section{Equivalence of both bounds}
\label{app:bounds_equiv}

\begin{theorem}
    The boundary in the energy-entropy space $(E,S)$ defined in equation (12) of Reference \cite{stilck_franca_limitations_2021} reading
    \begin{equation}\label{eq:boundaryNP}
    \mathrm{Tr}[\rho H]\geq \sup_{\beta>0}\beta^{-1}\left(-\ln(\mathrm{Tr}[e^{-\beta H}/2^n)-D(\rho||\mathbb{I}/2^n)\right)   
    \end{equation}
    is strictly equivalent to the family of Gibbs states  $\sigma_\beta=e^{-\beta H}/Z_\beta$, where $Z_\beta=\mathrm{Tr}[e^{-\beta H}]$, with free parameter the inverse temperature $\beta$ in Theorem 1 above. The right hand-side of equation (\ref{eq:boundaryNP}) being the union of the tangents to the boundary at each Gibbs state.
\end{theorem}

Before giving  a simple proof we need to provide two simple lemmas.
\begin{lemma}
    The expectation energy of a Gibbs state reads
    \begin{equation}
      U_{\beta}= \mathrm{Tr}[\sigma_\beta H]=-\frac{\partial \log Z_\beta}{\partial \beta}.
    \end{equation}
\end{lemma}
\begin{proof}
    Expanding the Gibbs state in the eigenbasis $\ket{e_k}$ of the Hamiltonian $H=\sum_{k}E_k\ketbra{e_k}{e_k}$ we obtain 
    \begin{equation}
       U_{\beta}=\sum_k E_k\frac{e^{-\beta E_k}}{Z_\beta}
       =-\sum_k \frac{\frac{\partial}{\partial\beta} e^{-\beta E_k}}{Z_\beta}=-\frac{\partial \log Z_\beta}{\partial \beta}.
    \end{equation}
\end{proof}

\begin{lemma}
    States of the Gibbs family $\sigma_\beta=e^{-\beta H}/Z_\beta$ with parameter $\beta$, expected energy $U_\beta=\mathrm{Tr}[\sigma_\beta H]$ and entropy  $S_\beta=-\mathrm{Tr}[\sigma_\beta\log\sigma_\beta]$ satisfy the relation
    \begin{equation}
        \beta \frac{\partial U_\beta}{\partial \beta}=\frac{\partial S_\beta}{\partial \beta}.
    \end{equation}
\end{lemma}

\begin{proof}
    We have seen in the proof of Theorem 1 that 
    $D(\sigma_{\beta}||\sigma_{\beta}) = \beta\mathrm{Tr}[\sigma_{\beta} H]+\log Z_{\beta}-S(\sigma_{\beta})=0$, which can be rewritten using our more compact notation  as
    $\beta U_\beta-S_\beta=-\log Z_\beta$. Taking the partial derivative over $\beta$ on both sides leads to 
    \begin{equation}
        U_\beta+\beta \frac{\partial U_\beta}{\partial \beta}-\frac{\partial S_\beta}{\partial \beta}
        =-\frac{\partial\log Z_\beta}{\partial \beta} = U_\beta,
    \end{equation}
    where the last equality results from Lemma 1. Eliminating $U_\beta$ both on the left and right sides concludes the proof.  
\end{proof}

We are now ready to prove Theorem 2. Remark that using the fact that 
$D(\rho||\mathbb{I}/2^n)=n-S(\rho)$, equation \ref{eq:boundaryNP} can be rewritten as
\begin{equation}\label{eq:bound2}
    \mathrm{Tr}[\rho H]\geq \sup_{\beta>0}\beta^{-1}\left(-\log Z_\beta +S(\rho)\right).   
\end{equation}
Then the boundary was obtained by computing the supremum over $\beta$ fixing the value $S(\rho)$.
    
\begin{proof}
    Consider a specific state of the Gibbs state family  $\sigma_\beta=e^{-\beta H}/Z_\beta$ of parameter $\beta$ with corresponding energy-entropy $(U_\beta,S_\beta)$. The set of energy-entropy values that are located above the tangent to the Gibbs state lower-bound at the point
    $(U_\beta,S_\beta)$ need to satisfy the constraint:
    \begin{equation}
        U-U_\beta \geq \frac{\frac{\partial U}{\partial \beta}}{\frac{\partial S}{\partial \beta}}\left(S-S_\beta\right). 
    \end{equation}
    Using Lemma 2 one can see that it is equivalent to the relation 
    \begin{equation}
        \beta U-S\geq \beta U_\beta - S_\beta =-\log Z_\beta,
    \end{equation}
    where the last equality results from Lemma 1. This last line can be rearranged as
    \begin{equation}
        U\geq \beta^{-1}\left(-\log Z_\beta+S\right),
    \end{equation}
    which, up to the optimization over $\beta$, is equation (\ref{eq:bound2}).
\end{proof}

\section{Lower-bounding of the Gibbs states energies}
\label{app:lemma_partition}
\begin{lemma}
    For any Hamiltonian decomposing into a sum of Hamiltonians $H_j$, $H=\sum_i H_j$, and for any quantum state $\rho$, the following holds:
    \begin{equation}
    \mathrm{Tr}(\rho H) \geq \sum_j\mathrm{Tr}(\sigma_{\beta_j}^{(j)} H_j)
\end{equation} 
where $\sigma_{\beta_j}^{(j)}$ stands for the Gibbs state respective to the Hamiltonian $H_j$ at temperature $\beta_j$ such that $S(\rho)=S(\sigma_{\beta_j}^{(j)})$.
\end{lemma}

\begin{proof}
    We prove the lemma for a two-Hamiltonian partition $H=H_A+H_B$, but it straightforwardly extends to any partition. Let $\rho$ be a quantum state. Due to the linearity of the trace, we have $\mathrm{Tr}(\rho H) = \mathrm{Tr}(\rho H_A) + \mathrm{Tr}(\rho H_B) $. We consider the Gibbs state $\sigma_{\beta_A}^A$ (resp. $\sigma_{\beta_B}^B$) with regards to $H_A$ (resp. $H_B$) such that
 \begin{equation}
     S(\sigma_{\beta_A}^A) = S(\sigma_{\beta_B}^B) = S(\rho).
 \end{equation} 
 We thus have $\mathrm{Tr}(\rho H_A) \geq \mathrm{Tr}(\sigma_{\beta_A}^A H_A)$ (resp. $\mathrm{Tr}(\rho H_B) \geq \mathrm{Tr}(\sigma_{\beta_B}^B H_B)$). All in all, we have:
\begin{equation}
\label{eq:lower_bounding}
    \mathrm{Tr}(\rho H) \geq \mathrm{Tr}(\sigma_{\beta_A}^A H_A) + \mathrm{Tr}(\sigma_{\beta_B}^B H_B).
\end{equation}
\end{proof}

\section{Energy and entropy associated to the Gibbs state of a Kronecker sum}
\label{app:kronecker_sums}
We consider a Hamiltonian $H$ which can be decomposed as a so-called Kronecker sum, namely 

\begin{align*}
    H &= \Tilde{H}_1 \otimes \mathds{1}_{d_2} \otimes \cdots \otimes \mathds{1}_{d_n} \\
    &+ \mathds{1}_{d_1} \otimes \Tilde{H}_2 \otimes \mathds{1}_{d_3} \otimes \cdots \otimes \mathds{1}_{d_n} \\
    & +\cdots \\
    & + \mathds{1}_{d_1} \otimes \cdots \otimes \mathds{1}_{d_{n-1}} \otimes \Tilde{H}_n \\
    & \overset{\mathrm{not.}}{\equiv} \bigoplus_j \Tilde{H}_j
\end{align*}
where each local Hamiltonian $\Tilde{H}_j$ is acting on a subspace of dimension $d_j$ of the full Hilbert space of $H$ of dimension $2^N$, such that $\prod_j d_j = 2^N$.

The Gibbs state $\sigma_{\beta}$ associated with $H$ at inverse temperature $\beta$ factorizes as

\begin{align}
    \sigma_{\beta} &= \frac{e^{-\beta \bigoplus_j \tilde{H}_j}}{\mathrm{Tr}\left( e^{-\beta \bigoplus_j \tilde{H}_j}\right)} \nonumber \\
    &= \frac{\bigotimes_j e^{-\beta \tilde{H}_j}}{\prod_j \mathrm{Tr}\left( e^{-\beta \tilde{H}_j}\right)} \nonumber\\
    &= \bigotimes_j \sigma_{\beta}^{(j)}.
\end{align}

Thus, the energies add up

\begin{align}
\label{eq:sum_energies}
    E(\sigma_{\beta}) &= \mathrm{Tr}\left( \bigotimes_j \sigma_{\beta}^{(j)} \bigoplus_j \tilde{H}_j \right) \nonumber \\
    &= \sum_j \mathrm{Tr}\left( \sigma_{\beta}^{(j)} \tilde{H}_j \right) \prod_{j' \neq j} \underbrace{\mathrm{Tr}\left( \sigma_{\beta}^{(j')} \right)}_{=1} \nonumber \\
    &= \sum_j \mathrm{Tr}\left( \sigma_{\beta}^{(j)} \tilde{H}_j \right) \nonumber \\
    &= \sum_j E\left(\sigma_{\beta}^{(j)}\right)
\end{align}

whereas

\begin{align}
\label{eq:multiply_part_func}
    Z_{\beta} &= \mathrm{Tr} \left( \bigotimes_j \sigma_{\beta}^{(j)}\right) \nonumber \\
    &= \prod_j \mathrm{Tr} \left( \sigma_{\beta}^{(j)} \right) \nonumber \\
    &= \prod_j Z_{\beta}^{(j)}.
\end{align}

As a consequence, the entropies add up just as the energies:

\begin{align}
\label{eq:sum_entropies}
    S(\sigma_{\beta}) &= \beta E(\sigma_{\beta}) + \log Z_{\beta} \nonumber \\
    &= \beta \sum_j E\left(\sigma_{\beta}^{(j)}\right) + \sum_j \log Z_{\beta}^{(j)} \nonumber \\
    &= \sum_j S\left( \sigma_{\beta}^{(j)} \right)
\end{align}

\section{Diagonalization of the tight-binding and atomic Hamiltonians for the first lower-bound}
\label{app:TB_atomic}

The $t$ term and $U$ term in Fermi-Hubbard are very different in nature: the first favours a wave-like behaviour whereas the second corresponds to localized particles. Both are tractable with the number of lattice sites. Thus we can group terms in the Hubbard Hamiltonian into two Hamiltonians covering very distinct physics:

\begin{equation}
    H = H_{\mathrm{TB}} + H_{\mathrm{at}}
\end{equation}

with

\begin{equation}
\label{eq:TB_hamilt}
    H_{\mathrm{TB}} = -t\sum_{\langle i, j \rangle \in E, \sigma=\uparrow, \downarrow} (c^{\dagger}_{i\sigma} c_{j \sigma} + \mathrm{h.c.}),
\end{equation}
the 'tight-binding part' and

\begin{equation}
\label{eq:atomic_hamilt}
    H_{\mathrm{at}} = - \mu \sum_{i \in V} (n_{i\uparrow} + n_{i \downarrow}) + U \sum_{i \in V}n_{i\uparrow}n_{i \downarrow}
\end{equation}
the interacting part. Note that we could as well have decided that the $\mu$ part belonged to the TB Hamiltonian. The diagonalization of each Hamiltonian is well known and expanded below.

\textit{Kinetic part} Within PBC the tight-binding Hamiltonian \ref{eq:TB_hamilt} is diagonal in the Fourier basis, reading
\begin{equation}
\label{eq:diagonal_kin_hamilt}
        H_{\mathrm{TB}} =  \sum_{k \sigma} \epsilon_{\bm{k}} f^{\dagger}_{\bm{k} \sigma} f_{\bm{k} \sigma}.    
\end{equation}
For a square lattice in 2d with lattice spacing $a$ we have: 

\begin{equation}
\label{eq:dispersion}
    \epsilon_{\bm{k}} = -2t (\cos{k_x a} + \cos{k_y a})
\end{equation}
with in each direction $x$ and $y$ $L$ momenta components of the form
\begin{equation}
    k_j = \frac{2\pi l}{L a}.
\end{equation}

Since \ref{eq:diagonal_kin_hamilt} represents a sum of non-interacting fermionic Hamiltonians $H_{\bm{k}} = \sum_{\sigma} \epsilon_{\bm{k}} f^{\dagger}_{\bm{k} \sigma} f_{\bm{k} \sigma}$, we can add up the Gibbs states energies associated to each mode to get the Gibbs state energy of the full kinetic Hamiltonian, as explicited in Appendix \ref{app:kronecker_sums}. Since the eigenvalues of $H_{\bm{k}}$ are $(0, \epsilon_{\bm{k}}, \epsilon_{\bm{k}}, 2\epsilon_{\bm{k}})$ corresponding to the mode $\bm{k}$ being respectively empty, occupied by a spin-up electron, occupied by a spin-down electron and doubly occupied, we readily get

\begin{equation}
    E(\sigma_{\beta}^{(\mathrm{TB})}) = 2 \sum_{\bm{k}} \epsilon_{\bm{k}} \frac{e^{-\beta \epsilon_{\bm{k}}}}{1 + e^{\beta \epsilon_{\bm{k}}}}
\end{equation}
which reflects the Fermi-Dirac prescription for the occupation of single-particle modes at inverse temperature $\beta$. 

The entropy reads
\begin{equation}
    S(\sigma_{\beta}^{(\mathrm{TB})}) = \beta E(\sigma_{\beta}^{(\mathrm{TB})}) + \sum_{\bm{k}}\log(1 + 2e^{-\beta \epsilon_{\bm{k}}} + e^{-2\beta \epsilon_{\bm{k}}}).
\end{equation}

\textit{Atomic part} The atomic part on the other hand is already diagonal in the site-spin basis. Each onsite Hamiltonian $H_i = -\mu(n_{i\uparrow} + n_{i \downarrow}) + Un_{i\uparrow}n_{i \downarrow}$ has eigenvalues $(0, -\mu, \mu, U-2\mu)$. Applying once again the results of Appendix \ref{app:kronecker_sums} we get

\begin{equation}
    E(\sigma_{\beta}^{(\mathrm{at})}) = N_{\mathrm{sites}}\frac{ \left( -2\mu e^{-\beta \mu} + (U-2\mu)e^{-\beta(U-2\mu)}\right)}{\left(  1 + 2e^{-\beta \mu} + e^{-\beta (U - 2 \mu)}\right)}.
\end{equation}
and
\begin{equation}
    S(\sigma_{\beta}^{(\mathrm{at})}) = \beta E(\sigma_{\beta}^{(\mathrm{at})}) + N_{\mathrm{sites}} \log\left(  1 + 2e^{-\beta \mu} + e^{-\beta (U - 2 \mu)} \right).
\end{equation}

\section{Lower-bounding based on 1D systems (\textit{One-dim.})}
\label{app:LB2}

For a square lattice with $L \times L$ sites, one can consider non interacting chains along the horizontal direction on the one hand and non interacting chains along the vertical direction on the other hand, as 

\begin{equation}
H = H_{\mathrm{h}} + H_{\mathrm{v}} 
\end{equation}
in which the horizontal part $H_{\mathrm{h}}$ (resp. the vertical part $H_{\mathrm{v}}$) only contains hopping $t$ along the horizontal (resp. vertical) edges and onsite interaction and chemical potential respectively set as $U/2$ and $\mu/2$. Note that we refer only abusively to this type of partitioning as geometric, as we are actually grouping the \textit{terms} of the Hamiltonian. 

Since the horizontal ($H_{\mathrm{h}} $) and the vertical ($H_{\mathrm{v}} $) parts of the Hamiltonian describe the same system upon inverting the $x$ and $y$ directions, the embodiment of \ref{eq:lower_bounding} reads
\begin{equation}
    \mathrm{Tr}(\rho H) \geq 2 \mathrm{Tr}(\sigma_{\beta_{\mathrm{v}}}^{\mathrm{v}} H_{\mathrm{v}}). 
\end{equation}

$H_{\mathrm{v}}$ can further be decomposed onto a Kronecker sum of 'fragment' Hamiltonians acting non-trivially only on a 1d portion of the lattice. The precise way these sums are written depends on the ordering convention for the fermionic modes. 

However, upon considering the ordering indexed by $2i + \sigma$ which starts at the upper-left site, snakes vertically through the lattice first considering spin up and then considering spin down, there are no fermionic statistics effects to be accounted for.

As a consequence, we have
\begin{equation}
    H_{\mathrm{v}} = \bigotimes H_{1d}
\end{equation}

with $H_{\mathrm{1d}}$ the Hamiltonian for a Fermi-Hubbard ring of $L$ sites with hopping $t$ and onsite interaction $U/2$ (and chemical potential $\mu/2$ if we decide to keep one). As a consequence, \ref{eq:sum_energies} and \ref{eq:sum_entropies} apply:

\begin{align}
    E(\sigma_{\beta}) &\geq 2LE\left(\sigma_{\beta'}^{(1d)}\right) \\
    S(\sigma_{\beta}) &= LS\left(\sigma_{\beta'}^{(1d)} \right)
\end{align}

With this decomposition, we can leverage the solvability of 1d Hubbard to compute the lower bound. 

\section{Lower-bounding based on plaquettes (\textit{Plaq.})}
\label{app:LB3}

For even $L \geq 4$ we can separate the Hamiltonian into two subhamiltonians which are Kronecker sums, and for which the summation property for the Gibbs states' energy and entropy applies (see Appendix \ref{app:kronecker_sums} for a derivation en these rules). \\

Let $\Lambda_{\bm{u}}(h_{\mathrm{plaquette}})$ be the operator centering the plaquette Hamiltonian $h_{\mathrm{plaquette}}$ at the location specified by $\bm{u}=(n_x, n_y)$ which we take to be the coordinates of the site sitting at the bottom left corner of the plaquette. Since we are considering PBC, the lattice sites' coordinates span $\llbracket 0, L-1\rrbracket \times \llbracket 0, L-1\rrbracket$ and plaquette $\mathcal{P}_{\bm{u}=(n_x, n_y)}$ covers sites at locations $(n_x, n_y), (n_x + 1 \; [2], n_y), (n_x, n_y + 1 \;[2]), (n_x + 1 \; [2], n_y+ 1 \; [2])$.

We define

\begin{equation}
    H_{\mathrm{even}}(U, t) = \sum_{\bm{u}=(0 \;[2], 0\;[2])}\Lambda_{\bm{u}}(h_{\mathrm{plaquette}}(U, t)),
\end{equation}
and similarly

\begin{equation}
    H_{\mathrm{odd}}(U, t) = \sum_{\bm{u}=(1 \;[2], 1\;[2])}\Lambda_{\bm{u}}(h_{\mathrm{plaquette}}(U, t)).
\end{equation}

Each sum comprises $\left( \frac{L}{2} \right)^2$ plaquettes which are non-overlapping. 
On the other hand, note that each edge is only present in one of $H_{\mathrm{even}}$ or $H_{\mathrm{odd}}$.

Thus,

\begin{equation}
    H = H_{\mathrm{even}}(U/2, t) + H_{\mathrm{odd}}(U/2, t)
\end{equation}
and it follows that

\begin{align}
    \mathrm{Tr}(\rho H) &= \mathrm{Tr}(\rho H_{\mathrm{even}}(U/2, t)) + \mathrm{Tr}(\rho H_{\mathrm{odd}}(U/2, t)) \nonumber\\
    &\geq \mathrm{Tr}(\sigma_{\beta_{\mathrm{even}}} H_{\mathrm{even}}(U/2, t)) + \mathrm{Tr}(\sigma_{\beta_{\mathrm{odd}}} H_{\mathrm{odd}}(U/2, t)) \nonumber\\
    &\geq 2 \mathrm{Tr}(\sigma_{\beta_{\mathrm{even}}} H_{\mathrm{even}}(U/2, t)) 
\end{align}

All in all, applying \ref{eq:sum_energies} and \ref{eq:sum_entropies} (choosing a site ordering accordingly in order to explicitly avoid fermionic statistics effects) we get:

\begin{align}
    E(\sigma_{\beta}) \geq\frac{L^2}{2} E(\sigma_{\beta'}^{{\mathrm{(plaquette)}}} (U/2, t)) \\
    S(\sigma_{\beta}) = \left(\frac{L}{2} \right)^2S(\sigma_{\beta'}^{{\mathrm{(plaquette)}}} (U/2, t)).
\end{align}

For odd $L$, we cannot find any tiling of the periodic lattice in terms of $n_{\mathrm{plaquettes}}$ plaquettes as we would require $4n_{\mathrm{plaquettes}}=L^2$, which cannot be satisfied unless $L$ is even. We could resort to a hybrid tiling with plaquettes, an open chain and non-interacting dimers rendering hoppings in between these structures but would loose the advantage of tractability. Another possibility would be to consider a tiling with $L^2$ subhamiltonians acting non-trivially on a single plaquette carrying Hubbard parameters $(U/4, t/2)$, but the obtained lower bound would be extremely loose due to the large number of unconstrained degrees of freedom in each of the  subhamiltonians.

\section{Comparison of lower-bound behaviour}
In addition to results presented in Figure \ref{fig:onset_scale_invariance} of the main text, we present in Figure \ref{fig:4sites_LBs} the plaquette lower-bound \textit{Plaq.}, which is intrinsically scale-invariant, as well as the two other lower-bounds \textit{Phenom.} and \textit{One-dim.} in the case $N_{\mathrm{sites}}=4$. In this case, the exact Gibbs states boundary can be computed via exact diagonalization and provides some insight into the tightness of each lower-bound. We plot results for values of correlation $U/t=0.1, 5, 10$. The almost uncorrelated case $U/t=0.1$ displays a global separation in terms of tightness of the lower-bounds on the whole range of entropy densities: as expected, the \textit{Phenom.} lower-bound is close to the true lower-bound since the tight-binding part of the Hamiltonian dictates the behaviour of the model. The second best lower-bound is \textit{One-dim.} and the loosest is \textit{Plaq.}. This is due to the fact that in the case $N_{\mathrm{sites}}=4$, \textit{Plaq.} is not a proper lower-bound since PBC are not relevant for a $2 \times 2$ lattice system. As $U/t$ increases, the \textit{One-dim.} lower-bound becomes better than \textit{Phenom.} on a wider and wider range of entropy densities. We expect both bounds to perform similarly when $U/t$ is large enough, since the density-density part of the Hamiltonian dominates the behaviour of the model.

\begin{figure}
    \centering
    \includegraphics[width=\linewidth]{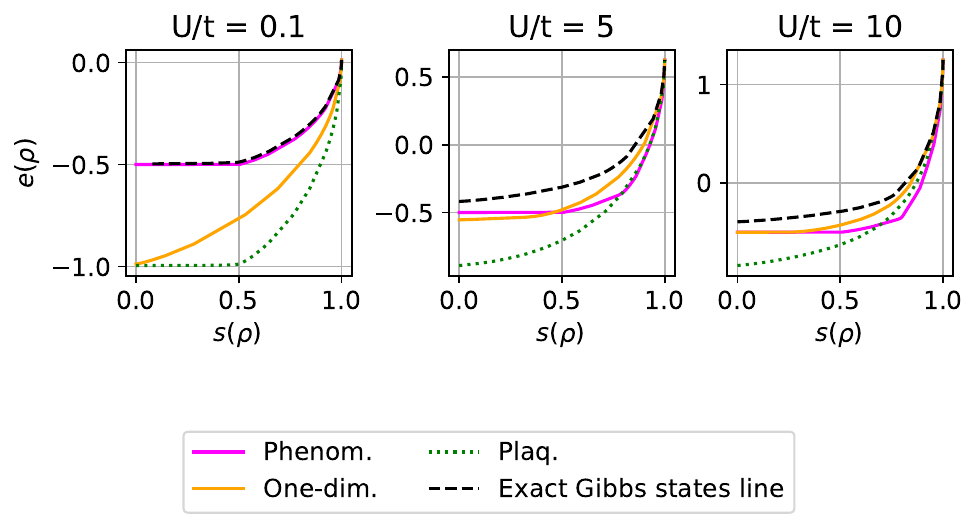}
    \caption{Comparison of the different lower bounds for the FHM with PBC on a square lattice with $N_{\mathrm{sites}}=4$, at different values of the correlation $U/t$. The scale-invariant \textit{Plaq.} lower-bound, only valid for even $L\geq 4$, is displayed with a dotted line on this figure but does not constitute a proper lower-bound as here $L=2$.}
    \label{fig:4sites_LBs}
\end{figure}
\section{Benchmarking for the Fermi-Hubbard model}
\label{app:benchmarking_details}

\subsection{Choice of setting}
\label{app_sub:choice_pb}

\paragraph{Problem size} The $L=8$ 2D FHM escapes exact diagonalization, but can be tackled with a near-term quantum computer as it translates into e.g. 128 qubits using Jordan-Wigner encoding. Although this value of $L$ may initially seem quite low considering we are interested in the physics of the FHM in the thermodynamic limit, there is evidence from the reported state-of-the-art classical ground state energy density that this is already quite close to the density in the thermodynamic limit since finite-size effects are already moderate (see \cite{qin_benchmark_2016}).

\paragraph{Filling} Classical results against which we wish to benchmark quantum optimization results are mostly obtained in the grand canonical ensemble and framed in terms of filling of the ground state $n_{\mathrm{f}}=\frac{1}{N_{\mathrm{sites}}}\bra{\psi_0}\hat{N} \ket{\psi_0}$. This filling reflects the value of the chemical potential $\mu$, but the relationship between these two values cannot be determined \textit{a priori} in general. As a consequence, classical variational methods typically proceed iteratively, tweaking $\mu$ until the target filling is obtained.
On the other hand, setting $\mu=U/2$ for the FHM without next-nearest neighbor hopping enforces particle-hole symmetry. As a consequence, the ground state is expected to be half-filled, meaning that it contains as many fermions as sites ($n_{\mathrm{f}}=1$). This makes half-filling a practical choice for us, since we can simply run the Gibbs states lower-bound estimation with $\mu=U/2$ and pick a classical comparison point displaying $n_{\mathrm{f}}=1$. 
On the flip side, half-filling renders Monte-Carlo methods such as AFQMC \cite{zhang_quantum_2003} (directly targeting the ground state) sign problem-free and as such, numerically exact. This means that we do not consider a proper candidate problem for quantum advantage. However, the quantum computer also benefits from half-filling in the sense that one can use a quantum circuit consisting in the preparation of a half-filled reference state (typically, the Hartree-Fock state) followed by a number-preserving unitary such as the plain-vanilla versions of LDCA and HVA referenced here.
The rationale of our benchmark is to determine whether current quantum devices could prepare states with a similar energy as found by classical means, as a first step towards quantum advantage in a setting escaping the reach of classical methods (e.g. away from half-filling). 

\subsection{Determination of the entropy threshold}
From Reference \cite{qin_benchmark_2016} we get $e_{\mathrm{class}}=-1.43$ from a AFQMC calculation. As illustrated on Figure \ref{fig:entropy_th_deter}, we deduce from the intersection to the Gibbs state lower bound with this value that $s_{\mathrm{th}}=0.69$.

\begin{figure}
    \centering
    \includegraphics[width=0.7\linewidth]{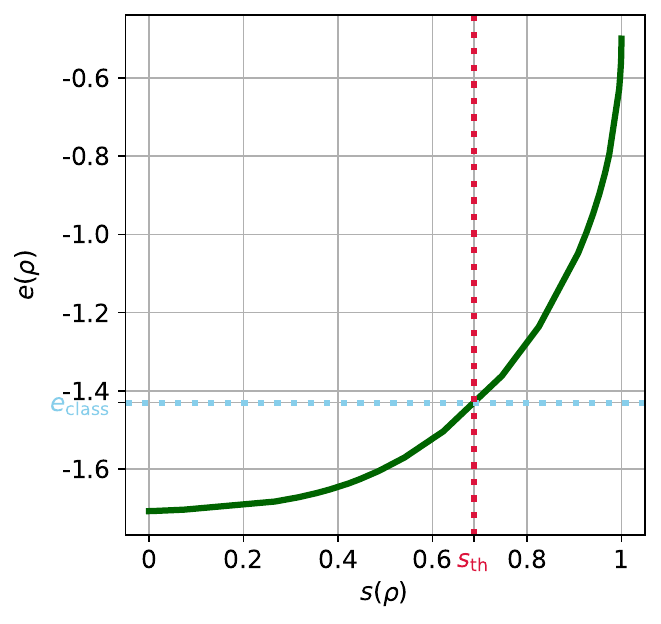}
    \caption{Determination of the entropy threshold for the half-filled 2D FHM with PBC at size $L=8$ and correlation $U/t=4$.}
    \label{fig:entropy_th_deter}
\end{figure}

\subsection{Derivation of the sub-threshold CNOT count}
Leveraging the suitability of a global depolarising noise model evidenced in Reference \cite{demarty_entropy_2024}, the CNOT count should not exceed

\begin{equation}
\label{eq:general_eq_max_depth}
    N_2 \leq \frac{\ln{\left(\frac{2^{nc} - 1}{2^n - 1}\right)}}{2\ln{(1-p_2)}}.
\end{equation}
with $n$ the number of qubits and

\begin{equation}
    c = 1 - s_{\mathrm{th}}.
\end{equation}

This stems from the following derivation based on a similar derivation in Reference \cite{demarty_entropy_2024}:

\begin{widetext}
\begin{align*}
    \frac{S^{(2)}}{n} &= -\frac{1}{n}\log_2 \left( \mathrm{Tr} \left(\rho_{\mathrm{out}}^2 \right) 
\right) \\
 &= -\frac{1}{n}\log_2 \left( 
 (1-2^{-n})(e^{2N_2 \ln(1-p_2)} - 1) + 1)
\right) \\
& \leq s_{\mathrm{th}} \equiv 1 -c \\
\Leftrightarrow \ln \left( 
 (1-2^{-n})(e^{2N_2 \ln(1-p_2)} - 1) + 1) \right)  &\geq -n \ln(2) (1-c) \\
 \Leftrightarrow (1-2^{-n})(e^{2N_2 \ln(1-p_2)} - 1) + 1 &\geq e^{-n \ln(2) (1-c)} \equiv 2^{-n(1-c)} \\
 \Leftrightarrow e^{2N_2 \ln(1-p_2)} - 1 &\geq \frac{2^{-n(1-c)} - 1}{1-2^{-n}} \\
  \Leftrightarrow e^{2N_2 \ln(1-p_2)}  &\geq \frac{2^{-n(1-c)} - 1 + 1 - 2^{-n}}{1-2^{-n}} \equiv \frac{2^{-n(1-c)} - 2^{-n}}{1-2^{-n}} = \frac{2^{nc}-1}{2^n - 1} \\
  \Leftrightarrow N_2 &\leq \frac{
  \ln \left(  \frac{2^{nc}-1}{2^n - 1} \right)
  }
  {
  2\ln(1-p_2)
  }
\end{align*}
\end{widetext}
\subsection{Maximumg number of layers for the Hamiltonian Variational Ansatz}
The idea behind the Hamiltonian Variational Ansatz is to consider the trotterized time evolution with $H=\sum_j H_j$ (where terms in each $H_j$ commute) over a reference state $\ket{\psi_{\mathrm{ref}}}$ (the easy-to-prepare ground state of some part of $H$, or the Hartree-Fock state)

\begin{equation}
    \ket{\psi(t)} = \left(\prod_{j}e^{-iH_j \frac{t}{N_T}} \right)^{N_T} \ket{\psi_{\mathrm{ref}}}.
\end{equation}
and turn it into an ansatz by replacing the timesteps by variational parameters:

\begin{equation}
    \ket{\psi(\bm{\theta})} = \prod_{k=1}^{N_{\mathrm{layers}}}\left(\prod_{j}e^{-iH_j \theta_j^{(k)}} \right)\ket{\psi_{\mathrm{ref}}}
\end{equation}

\paragraph{CNOT count}
The gate count depends on the fermion-to-qubit mapping, the connectivity of the hardware, the choice of reference state and the choice for the terms $H_j$. Typically we can start from the ground state of the quadratic part of the Hamiltonian as $\ket{\psi_{\mathrm{ref}}}$ and consider e.g. Jordan-Wigner encoded terms of the form 

\begin{align}
    n_{i\uparrow} n_{i \downarrow} &\rightarrow Z_{\alpha} Z_{\gamma} \\
    c^{\dagger}_{i\sigma} c_{j \sigma} + \mathrm{h.c.} &\rightarrow X_{\alpha} \left(\otimes_{k=\alpha + 1}^{\beta-1} Z \right) X_{\beta} +  Y_{\alpha} \left(\otimes_{k=\alpha + 1}^{\beta-1} Z \right) Y_{\beta} 
\end{align}

\begin{table}[]
    \centering
    \begin{tabular}{|c|c|c|}
    \hline 
    edge type & $|i-j|$ & \textnormal{count}  \\
        \hline 
    \textnormal{vertical} & 1 & L(L-1) \\ \hline 
    \textnormal{horizontal} & L & L(L-1) \\
    \hline 
    \textnormal{vertical PBC} & L & L \\
    \hline 
    \textnormal{horizontal PBC} & L(L-1) & L \\
     \hline
\end{tabular}
\caption{Count of hopping terms.}
\label{table:hop_terms_QAOA}
\end{table} 

There are also density terms corresponding to the chemical potential, but we do not spell their implementation out as they are single-qubit contributions. Time evolution associated to both types of terms ($e^{-i\theta H_j}$) can be exactly implemented. The two-qubit gate count will reflect SWAP gates as well as CNOT gate corresponding to two-qubit rotations. The two-qubit gate count for implementing a gate mimicking a hopping term between qubits $\alpha$ and $\beta$ is thus $2(|\alpha - \beta| - 1)$ SWAP gates \cite{Mineh:2022hyv} and 2 CNOT gates corresponding to the implementation of the nearest-neighbout $e^{i\theta (XX + YY)}$ gate. Since a SWAP gate decomposes into 3 CNOT gates, we obtain for the CNOT count $6|\alpha-\beta|- 4$. Let us now choose an ordering, say the columns ordering (we go on the graph column by column, downwards, to index sites). These indices $k$ will correspond to, say, the qubits corresponding to up spin-orbitals and indices $L^2 + k$ will label their down spin counterparts. Then we have, for each spin species, hopping terms reported in Table \ref{table:hop_terms_QAOA}  which yields a total CNOT gate count per layer of hopping terms

 \begin{equation}
     \#\mathrm{CNOT}_{\mathrm{hopping}}=2 \times 2L(6L^2 + 4L -3)
 \end{equation}

For the density-density terms, since we consider a qubit ordering with spin-up orbitals and then their spin-down counterparts in the same order, we always have $|i-j|=L^2$ so that
 \begin{equation}
\#\mathrm{CNOT}_{\mathrm{dens-dens}}= L^2 (6L^2 - 4L)
 \end{equation}. 

 All in all we have $\#\mathrm{CNOT}_{\mathrm{HVA \: layer}}=2 L (3L^3 + 12L^2-10L -6)$. We should add to that the cost of the reference state preparation. We can evaluate this CNOT count to be $n/2 \times 4(n-1) \times 2 = 8L^2(2L^2-1)$, reflecting the number of CNOT gates required for free-fermion state preparation based on mathchgates \cite{jozsa_matchgates_2008}.

 This translates into

\begin{equation}
\boxed{
    N^{(\mathrm{HVA})}_{\mathrm{L, max}} =
    \max\left(
    \left\lfloor
    \frac{
      \frac{
        \ln\left( \frac{2^{nc} - 1}{2^n - 1} \right)
      }{
        2 \ln\left(1 - p_2\right)
      }
      - 8L^2\left(2L^2 - 1\right)
    }{
      2L\left(3L^3 + 12L^2 - 10L - 6\right)
    }
    \right\rfloor,
    0
    \right)
}
\end{equation}

 \subsection{Maximum number of layers for the Low-Depth Circuit Ansatz}
LDCA \cite{dallaire-demers_low-depth_2019} is another physics-inspired variational ansatz for correlated ground states preparation, based on an exact circuit for free-fermion state preparation \cite{wecker_solving_2015}. The repeated pattern in the LDCA ansatz is a sequence of nearest-neighbour five two-qubit rotations (matchgates) applied first on pairs starting at even indices of qubits 0-1, 2-3, \dots, and then on odd pairs $1-2$, $3-4$, \dots. Such a sequence is repeated $N/2$ times to form a so-called LDCA cycle. We will refer to such cycles as layers to unify the terminology when comparing to the HVA ansatz.

The CNOT count of LDCA is thus $N_{\mathrm{L}} \times N/2 \times (N-1) \times 5 \times 2\simeq 5N_{\mathrm{L}}N^2 $ CNOT (for large $N$). As a consequence,

 \begin{equation}
\boxed{
    N^{(\mathrm{LDCA})}_{\mathrm{L, max}} = \left \lfloor \frac{
    \ln \left( \frac{2^{nc}-1}{2^n - 1} \right)
    }{40L^4 \ln(1-p_2)} \right \rfloor
}.
\end{equation}

\end{document}

%% file: fig1_svg-tex.pdf_tex
\begingroup%
  \makeatletter%
  \providecommand\color[2][]{%
    \errmessage{(Inkscape) Color is used for the text in Inkscape, but the package 'color.sty' is not loaded}%
    \renewcommand\color[2][]{}%
  }%
  \providecommand\transparent[1]{%
    \errmessage{(Inkscape) Transparency is used (non-zero) for the text in Inkscape, but the package 'transparent.sty' is not loaded}%
    \renewcommand\transparent[1]{}%
  }%
  \providecommand\rotatebox[2]{#2}%
  \newcommand*\fsize{\dimexpr\f@size pt\relax}%
  \newcommand*\lineheight[1]{\fontsize{\fsize}{#1\fsize}\selectfont}%
  \ifx\svgwidth\undefined%
    \setlength{\unitlength}{297.287bp}%
    \ifx\svgscale\undefined%
      \relax%
    \else%
      \setlength{\unitlength}{\unitlength * \real{\svgscale}}%
    \fi%
  \else%
    \setlength{\unitlength}{\svgwidth}%
  \fi%
  \global\let\svgwidth\undefined%
  \global\let\svgscale\undefined%
  \makeatother%
  \begin{picture}(1,0.99609384)%
    \lineheight{1}%
    \setlength\tabcolsep{0pt}%
    \put(0,0){\includegraphics[width=\unitlength,page=1]{fig1_svg-tex.pdf}}%
  \end{picture}%
\endgroup%